\documentclass[letterpaper,runningheads,11pt]{llncs}

\usepackage{amssymb,amsmath}
\usepackage{graphicx}
\usepackage{multicol}
\usepackage[colorlinks=true,hypertexnames=false,citecolor=blue,linkcolor=blue,filecolor=blue]{hyperref}
\usepackage{layout}

\newenvironment{proposition-repeat}[1]{\begin{trivlist}
\item[\hspace{\labelsep}{\bf\noindent Proposition~\ref{#1} }]}%
{\end{trivlist}}

\newenvironment{corollary-repeat}[1]{\begin{trivlist}
\item[\hspace{\labelsep}{\bf\noindent Corollary~\ref{#1} }]}%
{\end{trivlist}}

\newenvironment{lemma-repeat}[1]{\begin{trivlist}
\item[\hspace{\labelsep}{\bf\noindent Lemma~\ref{#1} }]}%
{\end{trivlist}}

\newenvironment{theorem-repeat}[1]{\begin{trivlist}
\item[\hspace{\labelsep}{\bf\noindent Theorem~\ref{#1} }]}%
{\end{trivlist}}

\title{The read/write protocol complex is collapsible\thanks{A short 
version  to appear in the Springer LNCS proceedings of LATIN 2016. $\,$ Partially supported by UNAM-PAPIIT grant.}}
\author{Fernando Benavides\inst{1,2} \and Sergio Rajsbaum\inst{1}}
\institute{Instituto de Matem\'aticas, Universidad Nacional Aut\'onoma de M\'exico,\\ Ciudad Universitaria, D.F. 04510 M\'exico
\\ \email{rajsbaum@im.unam.mx}
\and
Departamento de Matem\'aticas y Estad\'istica, Universidad de Nari\~no, \\ San Juan de Pasto, Colombia
\\ \email{fandresbenavides@gmail.com}
}

\authorrunning{Fernando Benavides and Sergio Rajsbaum}

\begin{document}
\maketitle
\begin{abstract}
The celebrated \emph{asynchronous computability theorem}  provides a characterization of the class of decision tasks that can be solved in a wait-free manner by asynchronous processes that communicate by writing and taking atomic snapshots  of a shared memory. Several variations of the  model have been proposed (immediate snapshots and iterated immediate snapshots),  all  equivalent for wait-free solution of decision tasks, in spite of the fact that the protocol complexes that arise from the different  models are structurally distinct. The topological and combinatorial properties of these snapshot protocol complexes have been studied in detail, providing  explanations for why the asynchronous computability theorem holds in all the  models.

In reality concurrent systems do not provide processes with snapshot operations. Instead, snapshots are implemented (by a wait-free protocol) using operations that write and read individual shared memory locations. Thus, read/write protocols are also computationally equivalent to snapshot protocols. However, the structure of the read/write protocol complex has not been studied. In this paper we show that the read/write iterated protocol complex is collapsible (and hence contractible). Furthermore, we show that a distributed protocol that wait-free implements atomic snapshots in effect is performing the collapses.
\end{abstract}

\section{Introduction}

A  \emph{decision task} is the distributed equivalent of a function, where each process knows only part of the input, and after communicating with the other processes, each process computes part of the output. For example, in the $k$-\emph{set agreement} task processes have to agree on at most $k$ of their input values; when $k=1$ we get the  \emph{consensus} task~\cite{FischerLP85}.

A central concern in distributed computability is studying which tasks are solvable in a distributed computing model, as determined by the type of communication mechanism available  and the reliability of the processes. Early on it was shown that consensus is not solvable even if only one process can fail by crashing,
when asynchronous processes communicate by message passing~\cite{FischerLP85} or even by writing and reading a shared memory~\cite{LouiAA:87}. A graph theoretic characterization of the tasks solvable in the presence of at most one process failure appeared soon after~\cite{BiranMZ90}.

The  \emph{asynchronous computability theorem}~\cite{1999TopologicalStructureAsynchronous_HS} exposed that moving from tolerating one process  failure, to any number of process failures, yields a characterization of the class of decision tasks that can be solved in a wait-free manner by asynchronous processes based on simplicial complexes, which are higher dimensional versions of graphs. In particular, $n$-set agreement is not wait-free solvable, even for  $n+1$ processes~\cite{1993GeneralizedFLPImposibility_BG,1999TopologicalStructureAsynchronous_HS,SaksZ00}.

Computability theory through combinatorial topology has evolved  to encompass arbitrary malicious failures, synchronous and partially synchronous processes, and various communication mechanisms~\cite{2013DistributedCombinatorialTopology_HKR}. Still,  the original wait-free model of the asynchronous computability theorem, where crash-prone processes that communicate wait-free by writing and reading a shared memory is  fundamental. For instance, the question of solvability in other models (e.g. $f$ crash failures), can in many cases be reduced to the question of wait-free solvability~\cite{BorowskyGLR01,HerlihyR12}.

More specifically, in the \emph{AS model} of~\cite{2013DistributedCombinatorialTopology_HKR} each process can write its own location of the shared-memory, and it is able to  read the whole shared memory in one  atomic step, called a \emph{snapshot}. The characterization  is based on the \emph{protocol complex,} which is a geometric representation of the various possible executions of a protocol. Simpler variations of this  model have been considered. In the \emph{immediate snapshot} (IS) version~\cite{AttiyaR2002,1993GeneralizedFLPImposibility_BG,SaksZ00}, proceses can execute  a combined write-snapshot operation. The \emph{iterated immediate snapshot} (IIS) model~\cite{1997SimpleAlgorithmicallyReasoned_BG} is even simpler to analyze, and can be extended (IRIS) to analyze partially synchronous models~\cite{2008IteratedRestrictedImmdediate_RRT}. Processes communicate by accessing a sequence of shared arrays, through immediate snapshot operations, one such operation in  each array. The success of the entire approach hinges on the fact that the topology of the protocol complex of a model determines critical information about the solvability of the task and, if solvable, about the complexity of solution~\cite{HoestS97}.

All these  snapshot models, AS, IS,  IIS and IRIS can solve exactly the same set of tasks. However, the protocol complexes that arise from the different  models are structurally distinct. The combinatortial topology  properties of these  complexes have been studied in detail, providing  insights for why some
tasks are solvable and others are not in a  model.
%all the snapshot models solve the same tasks.

\paragraph{\bf Results}
In reality concurrent systems do not provide processes with  snapshot operations. Instead, snapshots are implemented (by a wait-free protocol) using operations that write and read individual shared memory locations~\cite{1990AtomicSnapshotSharedMemory_AADGM}. Thus, read/write protocols are also computationally equivalent to snapshot protocols. However, the structure of the read/write protocol complex has not been studied. Our results are the following.
\begin{enumerate}
\item The one-round read/write protocol complex is collapsible to the IS protocol, i.e. to a chromatic subdivision of the input complex. The collapses can be performed simultaneously in entire orbits of the natural symmetric group action. We use ideas from~\cite{2013TopologyComplexView_K}, together with distributed computing techniques of partial orders.
\item Furthermore, the distributed protocol that wait-free implements  immediate snapshots of~\cite{1993ImmediateAtomicSnapshot_BG,2010RecursionDictributedComputing_GR} in effect is performing the collapses.
\item Finally,  also the multi-round iterated read/write complex is collapsible. We use ideas from~\cite{2014IteratedChromaticCollapsible_GMT}, together with carrier maps e.g.~\cite{2013DistributedCombinatorialTopology_HKR}.
\end{enumerate}
All omitted proofs are  in the Appendix.
\paragraph{\bf Related work}
The one-round immediate snapshot protocol complex is the simplest, with an elegant combinatorial representation; it is a chromatic subdivision of the input complex~\cite{2013DistributedCombinatorialTopology_HKR,2012ChromaticSubdivision_K}, and so is the (multi-round) IIS protocol~\cite{1997SimpleAlgorithmicallyReasoned_BG}. The multi-round (single shared memory array) IS protocol complex is harder to analyze, combinatorially it can be shown to be a pseudomanifold~\cite{AttiyaR2002}. IS and IIS protocols are homeomorphic to the input complex. An AS protocol complex is not generally homeomorphic to the underlying input complex, but it is homotopy equivalent to it~\cite{2004NoteHomotyTypeASComplex_Havlicek}. The span of~\cite{1999TopologicalStructureAsynchronous_HS} provides an homotopy equivalence  of the (multi-round) AS protocol complex to the input complex~\cite{2004NoteHomotyTypeASComplex_Havlicek},  clarifying the basis of the obstruction method~\cite{2000CompObstructions_Havlicek} for detecting impossibility of solution of tasks.

Later on stronger results were proved, about the collapsibility of the protocol complex. The one-round IS protocol complex is collapsible~\cite{2014TopologyISComplex_K} and homeomorphic to closed balls. The structure of the AS is more complicated, it was known to be contractible~\cite{2004NoteHomotyTypeASComplex_Havlicek,2013DistributedCombinatorialTopology_HKR}, and then shown to be collapsible (one-round) to the IS complex~\cite{2013TopologyComplexView_K}. The IIS (multi-round) version was shown to be collapsible too~\cite{2014IteratedChromaticCollapsible_GMT}.

There are several wait-free implementations of atomic snapshots starting with~\cite{1990AtomicSnapshotSharedMemory_AADGM}, but we are aware of only  two algorithms that implement  immediate snapshots; the original of~\cite{1993ImmediateAtomicSnapshot_BG}, and its recursive version~\cite{2010RecursionDictributedComputing_GR}.

\section{Preliminaries}

\subsection{Distributed computing model}\label{subsec:model}

The basic model we consider is the one-round  {\it read/write} model ($\mathsf{WR}$), e.g.~\cite{HerlihyShavitBook2008}. It consists of $n+1$ processes  denoted by the numbers $[n]=\{0,1,\ldots,n\}$. A process is a deterministic (possibly infinite) state machine. Processes  communicate through a shared memory array $\mathsf{mem}[0\ldots n]$ which consists of $n+1$ single-writer/multi-reader atomic registers. Each process  accesses the shared memory by invoking the atomic operations $\mathsf{write}(x)$ or $\mathsf{read}(j)$, $0\leq j\leq n$. The $\mathsf{write}(x)$ operation is used by process $i$  to write  value $x$ to register $i$, and process $i$  can invoke $\mathsf{read}(j)$ to read register $\mathsf{mem}[j]$, for any $0\leq j\leq n$. Each process $i$ has an input value, which may be its own id $i$. In its first operation, process $i$ writes its input  to $\mathsf{mem}[i]$, then it reads each of the $n+1$ registers, in an arbitrary order.
Such a sequence of operations, consisting of a write followed by all the reads is abbreviated by  $\mathsf{WScan}(x)$.

An \emph{execution} consists of an interleaving of the operations of the processes, and we assume any interleaving of the operations is a possible execution. We
may also consider an execution where only a subset of processes participate, consisting of an interleaving of the operations of those processes. These assumptions represent a wait-free model where any number of processes may fail by crashing.

In more detail, an execution is described as a set of atomic operations together with the irreflexive and transitive partial order given by: $op$ precedes $op'$ if $op$ was completed before $op'$. If $op$ does not precede $op'$ and viceversa, the operations are called {\it concurrent}. The set of  values read in an execution $\alpha$ by  process $i$ is called the {\it local view} of $i$ which is denoted by $view(i,\alpha)$. It consists of pairs $(j,v)$, indicating that the value $v$ was read from the $j$-th register. The set of all local views in the execution $\alpha$ is called the {\it view} of $\alpha$ and it is denoted by $view(\alpha)$. Let $\mathcal{E}$ be a set of executions of the WR model.  Consider the equivalence relation on $\mathcal{E}$ given by: $\alpha\sim \alpha'$ if  $view(\alpha)=view(\alpha')$. Notice that for every execution $\alpha$ there exists an equivalent \emph{sequential} execution $\alpha'$ with no concurrent operations. In other words, if $op$ and $op'$ are concurrent operations in $\alpha$ we can suppose that $op$ was executed immediately before $op'$ without modifying any views. Thus, we often consider only sequential  executions $\alpha$, consisting of  a linear order on the set of all write and read operations.

Two other models can be derived from the WR model. In the \emph{iterated} WR model, processes communicate through a sequence of arrays. They all go through the sequence of arrays $\mathsf{mem}_0$, $\mathsf{mem}_1\ldots$ in the same order, and in the $r$-th round, they access the $r$-th array, $\mathsf{mem}_r$, exactly as in the one-round version of the WR model. Namely,  process $i$ executes one write to $\mathsf{mem}_r[i]$  and then reads one by one all entries $j$, $\mathsf{mem}_r[j]$, in arbitrary order. In the \emph{non-iterated, multi-round} version of the WR model, there is only one array $\mathsf{mem}$, but processes can execute several rounds of writing and then reading one by one the entries of the array. The \emph{immediate snapshot} model (IS)~\cite{1993GeneralizedFLPImposibility_BG,SaksZ00}, consists of a subset of executions of the WR one round model. Namely, all the executions where the operations can be organized in \emph{concurrency classes}, each one consisting a set of writes  by the set of processes participating in the concurrency class, followed by a read to all registers by each of these processes. See Section~\ref{sec:addPropExec}.

\subsection{Algorithm IS}\label{subsec:Algo_IS}

Consider the recursive algorithm IS of~\cite{2010RecursionDictributedComputing_GR} for the iterated WR model, presented in Figure \ref{alg:IS}. Processes go trough a series of disjoint shared memory arrays $\mathsf{mem}_0,\mathsf{mem}_1,\ldots,\mathsf{mem}_n$. Each array $\mathsf{mem}_k$ is accessed by process $i$ invoking $\mathsf{WScan}(i)$ in the recursive call $\mathrm{IS}(n+1-k)$. Process $i$ executes  $\mathrm{WScan}(i)$ (line $(1)$), by performing first $\mathsf{write}(i)$,  followed by $\mathsf{read}(j)$ for each $j\in [ n]$, in an arbitrary order. The set of  values read (each one with its location) is what the invocation of $\mathsf{WScan}(i)$ returns. In  line $(2)$ the process checks if $view$ contains $n+1-k$ id's, else $\mathrm{IS}(n-k)$ is again invoked on the next shared memory in  line $(3)$. It is important to note that in each recursive call $\mathrm{IS}(n+1-k)$ at least one process returns with $|view|=n+1-k$, given that $n+1-k$ processes  invoked $\mathrm{IS}$. For example, in the first recursive call $\mathrm{IS}(n+1)$ the last process to write reads $n+1$ values and terminates the algorithm.

\begin{figure}[h]
\begin{center}
\fbox{\parbox{5.5cm}{
{\bf Algorithm} IS($n+1$)\\
(1) $view\leftarrow \mathsf{WScan}(i)$\\
(2) {\bf if} $|view|=n+1$ {\bf then} return $view$\\
(3) {\bf else} return IS($n$).
}}
\end{center}
\caption{Code for process $i$}
\label{alg:IS}
\end{figure}

\vspace{-0.3cm}

Every execution of the $\mathrm{IS}$ protocol  can be represented by a finite sequence $\alpha=\alpha_0,\alpha_1,\ldots,\alpha_l$ with $\alpha_k$ an execution of the WR one round model where every process that takes a step in $\alpha_k$ invokes the recursive call with $\mathrm{IS}(n+1-k)$. Since at least one process terminates the algorithm the length $l(\alpha)=l+1$ is at most $n+1$. The last returned local view in execution $\alpha$  for process $i$ is denoted $view(i,\alpha)$, and the set of all local views is  denoted by  $view(\alpha)$.

Denote by $\mathcal{E}_l$ the set of views of all executions $\alpha$ with $l(\alpha)=l+1$. Then $\mathcal{E}_n\subseteq\cdots\subseteq\mathcal{E}_0$. In particular, $\mathcal{E}_0$ corresponds to the views of executions of the one round WR of
 Section~\ref{subsec:model}.
 Also, $\mathcal{E}_n$ corresponds to the views of the immediate snapshot model, see Theorem $1$ of \cite{2010RecursionDictributedComputing_GR}.

\section{Definition and properties of the protocol complex}
\label{sec:Protocol_Complex}

Here we define the protocol complex of the write/read model and other models, which arise from the sets $\mathcal{E}_i$ mentioned in the previous section.

\subsection{Additional properties about executions}
\label{sec:addPropExec}

Recall from Section~\ref{subsec:model} that an execution  can be seen as a linear order on the set of write and read operations. For a subset $I\subseteq[n]$ let
\begin{equation*}
\mathcal{O}_I=\{w_i,r_i(j) \ : \ i\in I, \ j\in[n]\}.
\end{equation*}
with $I=\mathcal{O}_i=\emptyset$. A \emph{$wr$-execution on $I$} is a pair $\alpha=(\mathcal{O}_I,\rightarrow_{\alpha})$ with $\rightarrow_{\alpha}$ a linear order on $\mathcal{O}_I$ such that $w_i\rightarrow_{\alpha}r_i(j)$ for all $j\in[n]$. The set $I$ is called the $\mathtt{Id}$ set of $\alpha$ which is denoted by $\mathtt{Id}(\alpha)$. Hence the view of $i$ is $view(i,\alpha)=\{j\in I \ : \ w_j\rightarrow_{\alpha}r_i(j)\}$ and the view of $\alpha$ is $view(\alpha)=\{(i,view(i,\alpha)) \ : \ i\in I\}$. Note the chain $w_i\rightarrow_{\alpha}r_i(j_0)\rightarrow_{\alpha}\cdots\rightarrow_{\alpha}r_i(j_n)$ represents the invoking of $\mathsf{WScan}$ by the process $i$ in the $wr$-execution $\alpha$. Consider a $wr$-execution $\alpha$ and suppose that the order in which the process $i$ reads the array $\mathsf{mem}[0\ldots n]$ is given by $r_i(j_0)\rightarrow_{\alpha}\cdots\rightarrow_{\alpha}r_i(j_n)$. If every write operation $w_k$ satisfies $w_k\rightarrow_{\alpha} r_i(j_0)$ or $r_i(j_n)\rightarrow_{\alpha} w_k$ then $view(i,\alpha)$ corresponds to an atomic snapshot.

As a consequence, every execution in the snapshot model and immediate snapshot model corresponds to an execution in the write/read model. For instance in the $wr$-execution
\begin{multline*}
\alpha:w_2\rightarrow r_2(0)\rightarrow w_0\rightarrow r_0(0)\rightarrow r_0(1)\rightarrow r_0(2)\rightarrow w_1\rightarrow\\
\rightarrow r_1(0)\rightarrow r_2(1)\rightarrow r_1(1)\rightarrow r_2(2)\rightarrow r_1(2)
\end{multline*}
the $view(0,\alpha)=\{0,2\}$ and $view(1,\alpha)=[2]$ are immediate snapshots, this means the processes $0$ and $2$ could have read the array instantaneously.
In contrast,  the $view(2,\alpha)=\{1,2\}$ does not correspond to a snapshot.
For the following consider the process $j$ such that $w_i\rightarrow_{\alpha}w_j$ for all $i$.

\begin{proposition}\label{prop:one_process_terminates}
Let $\alpha$ be a $wr$-execution on $I$. Then there exist $j\in I$ such that $view(j,\alpha)=I$.
\end{proposition}

Let $\alpha$ be a $wr$-execution. For $0\leq k\leq n$, define $\mathtt{Id}_k(\alpha)=\{j\in\mathtt{Id}(\alpha) \ : \ |view(j,\alpha)|=n+1-k\}$. An \emph{IS-execution} is a finite sequence $\alpha=\alpha_0,\ldots,\alpha_l$ such that $\alpha_0$ is a $wr$-execution on $[n]$, and $\alpha_{k+1}$ is a $wr$-execution on $\mathtt{Id}(\alpha_k)-\mathtt{Id}_{k}(\alpha_k)$. Given an $IS$-execution $\alpha$, Proposition~\ref{prop:one_process_terminates} implies $l(\alpha)\leq n+1$. Moreover $\mathtt{Id}(\alpha_{k+1})\subseteq\mathtt{Id}(\alpha_{k})$ for all $0\leq k\leq l-1$. Hence $|\mathtt{Id}(\alpha_k)|\leq n+1-k$. Executions $\alpha,\alpha'$ are \emph{equivalent} if $view(\alpha)=view(\alpha')$, denoted $\alpha\sim\alpha'$.

\begin{lemma}\label{lem:equal_interval}
Let $\alpha$ and $\alpha'$ be $IS$-executions with $l(\alpha)=l(\alpha)$. Given $0\leq k\leq l$, (1) If $\alpha\sim\alpha'$ then $\mathtt{Id}(\alpha_k)=\mathtt{Id}(\alpha'_k)$. (2) If $\alpha_k\sim\alpha'_k$ then $\alpha\sim\alpha'$.
\end{lemma}

According to the behavior of the protocol in Figure~\ref{alg:IS}, the local view of $i$ is defined as $view(i,\alpha)=view(i,\alpha_k)$, if $i\in\mathtt{Id}(\alpha_k)-\mathtt{Id}(\alpha_{k+1})$ and $view(i,\alpha)=view(i,\alpha_l)$ for $k=l$. Hence the view of $\alpha$ is defined as
$view(\alpha)=\{(i,view(i,\alpha)) \ : \ i\in[n]\}.$

\begin{lemma}\label{lem:E_l+1_E_l}
Let $\alpha=\alpha_0,\ldots,\alpha_{l+1}$ be an $IS$-execution, $l(\alpha)=l+2$. Then $view(\alpha)=view(\alpha')$ for some  $IS$-execution $\alpha'$ such that $l(\alpha')=l+1$.
\end{lemma}

The $wr$-execution 
$\alpha'=\alpha_0,\ldots,\alpha_{l-1},\alpha'_{l}$  of the lemma is obtained by,
$\alpha'_{l}$ such that
\begin{equation*}
view(i,\alpha'_l)=\left\{
  \begin{array}{ll}
    view(i,\alpha_{l}), & \hbox{if} \ i\in\mathtt{Id}_{l}(\alpha_l) \\
    view(i,\alpha_{l+1}), & \hbox{if} \ i\not\in\mathtt{Id}_l(\alpha_l).
  \end{array}
\right.
\end{equation*}

It  follows $\mathcal{E}_l=\{ view(\alpha) \ : \ \alpha=\alpha_0,\ldots,\alpha_l \}$. Thus, Lemma~\ref{lem:E_l+1_E_l} implies  $\mathcal{E}_{l+1}\subseteq\mathcal{E}_l$. For example consider the $IS$-execution $\alpha=\alpha_0,\alpha_1,\alpha_2$ where
\begin{flushleft}
$\alpha_0:w_0\rightarrow r_0(0)\rightarrow r_0(1)\rightarrow r_0(2)\rightarrow w_1 \rightarrow r_1(0)\rightarrow r_1(1)\rightarrow r_1(2)
\rightarrow w_2\rightarrow r_2(0)\rightarrow r_2(1)\rightarrow r_2(2)$.
\end{flushleft}
\vspace{-0.4cm}
\begin{flushleft}
$\alpha_1:w_0\rightarrow r_0(0)\rightarrow r_0(1)\rightarrow r_0(2)\rightarrow w_1 \rightarrow r_1(0)\rightarrow r_1(1)\rightarrow r_1(2)$.
\end{flushleft}
\vspace{-0.4cm}
\begin{flushleft}
$\alpha_2:w_0\rightarrow r_0(0)\rightarrow r_0(1)\rightarrow r_0(2)$.
\end{flushleft}

\noindent So $view(\alpha)=\{(0,\{0\}),(1,\{0,1\}),(2,\{0,1,2\})\}\in\mathcal{E}_2\subseteq\mathcal{E}_1\subseteq\mathcal{E}_0$, Fig.~\ref{fig:Protocol_Complex_1_2},~\ref{fig:complexes_algo}.

\subsection{Topological definitions}

The following are standard technical definitions, see~\cite{2008SimplicialComplexesGraphs_J,2013TopologyComplexView_K}. A (abstract) \emph{simplicial complex}
$\Delta$ on a finite set $V$ is a collection of subsets of $V$ such that for any $v\in V$, $\{v\}\in\Delta$, and if $\sigma\in\Delta$ and $\tau\subseteq\sigma$ then $\tau\in\Delta$. The elements of $V$ are called \emph{vertices} and the elements of $\Delta$ \emph{simplices.} The \emph{dimension} of a simplex $\sigma$ is $\dim(\sigma)=|\sigma|-1$.  For instance the vertices are $0$-simplices. For the purposes of this paper, we adopt the convention that the {\it void complex} $\Delta=\emptyset$ is a simplicial complex which is different from the {\it empty complex} $\Delta=\{\emptyset\}$. Given a positive integer $n$, $\Delta^{n}$ denotes the standard simplicial complex whose vertex set is $[n]$ and every subset of $[n]$ is a simplex. From now on we identify a complex $\Delta$ with its collection of subsets. For every simplex $\tau$ we denote by $I(\tau)$ the set of all simplices $\rho$, $\tau\subseteq\rho$. A simplex $\tau$ of $\Delta$ is called \emph{free} if there exists a maximal simplex $\sigma$ such that $\tau\subseteq\sigma$ and no other maximal simplex contains $\tau$. The procedure of removing every simplex of $I(\tau)$ from $\Delta$ is called a \emph{collapse.}

Let $\Delta$ and $\Gamma$ be simplicial complexes, \emph{$\Delta$ is collapsible to $\Gamma$} if there exists a sequence of collapses leading from $\Delta$ to $\Gamma$. The corresponding procedure is denoted by $\Delta\searrow\Gamma$. In particular, if the collapse is elementary with free simplex $\tau$, it is denoted by $\Delta\searrow_{\tau}\Gamma$. If $\Gamma$ is the void complex,  $\Delta$ is \emph{collapsible}. The next definition from~\cite{2013TopologyComplexView_K} gives a  procedure to collapse a simplicial complex, by collapsing simultaneously by entire orbits of the group action on the vertex set. Let $\Delta$ be a simplicial complex with a simplicial action of a finite group $G$. A simplex $\tau$ is called $G$-free if it is free and for all $g\in G$ such that  $g(\tau)\neq\tau$, $I(\tau)\cap I(g(\tau))=\emptyset$. If $\tau$ is $G$-free, the procedure of removing every simplex $\rho\in\bigcup\limits_{g\in G} I(g(\tau))$ is called a $G$-collapse of $\Delta$.

Since, if $\tau$ is $G$-free then $g(\tau)$ is free as well, the above definition guarantees that all collapses in the orbit of $\tau$ can be done in any order {\it i.e} every $G$-collapse is a collapse. A simplicial complex $\Delta$ is \emph{$G$-collapsible to $\Gamma$} if there exist a sequence of $G$-collapses leading from $\Delta$ to $\Gamma$, it is denoted by $\Delta\searrow_{G}\Gamma$. In similar way, if the $G$-collapse is elementary with $G$-free simplex $\tau$, the notation $\Delta\searrow_{G(\tau)}\Gamma$ will be used. In the case $\Gamma$ is the void complex, $\Delta$ is called $G$-collapsible. For instance consider a $2$-simplex $\sigma$, $\tau$ a $1$-face of $\sigma$ and the action of $\mathbb{Z}_3$ over $\sigma$, then $\tau$ is free but not $\mathbb{Z}_3$-free however $\sigma$ is $\mathbb{Z}_3$-collapsible.

\subsection{Protocol Complex}

Let $n$ be a positive integer. The abstract simplicial complex $\mathsf{WR}_l(\Delta^n)$ with $0\leq l\leq n$ consists of the set of vertices
$V=\{(i,view_i) \ : \ i\in view_i\subseteq [n]\}$. A subset $\sigma\subseteq V$ forms a simplex if only if there exist an $IS$-execution $\alpha$ of length $l+1$ such that $\sigma\subseteq view(\alpha)$.

The complex $\mathsf{WR}_0(\Delta^n)$ is called the \emph{protocol complex} of the write/read model and it will be denoted by $\mathsf{WR}(\Delta^n)$. Protocol complexes for the particular cases $n=1$ and $n=2$ are shown in Figure~\ref{fig:Protocol_Complex_1_2}. In \cite{2013TopologyComplexView_K}  a combinatorial description of the protocol complex $\mathrm{View}^n$ associated to the snapshot model is given.  There every simplex of $\mathrm{View}^n$ is represented as a $2\times t$ matrix. Every simplex $\sigma\in\mathsf{WR}(\Delta^n)$ can be expressed as
\begin{equation*}
\mathrm{W}(\sigma)=\left(
                     \begin{array}{cccc}
                       V_1 & \cdots & V_{t-1} & [n] \\
                       I_1 & \cdots & I_{t-1} & I_t \\
                     \end{array}
                   \right)
\end{equation*}
where $I_i\cap I_j=\emptyset$ with $i\neq j$ and $I_i\subseteq  V_j$ for all $i\leq j$. Moreover we can associate a matrix for every simplex in the complex $\mathsf{WR}_l(\Delta^n)$. This representation implies that $\chi(\Delta^n)$ and $\mathrm{View}^n$ are subcomplexes of $\mathsf{WR}(\Delta^n)$. Figure~\ref{fig:complexes_algo} shows the complex $\mathsf{WR}_l(\Delta^2)$. From now on we will write $\mathsf{WR}_l$ instead of $\mathsf{WR}_l(\Delta^^n)$ unless we specify the standard complex.
\begin{figure}[h]
\begin{center}
\includegraphics[scale=0.7]{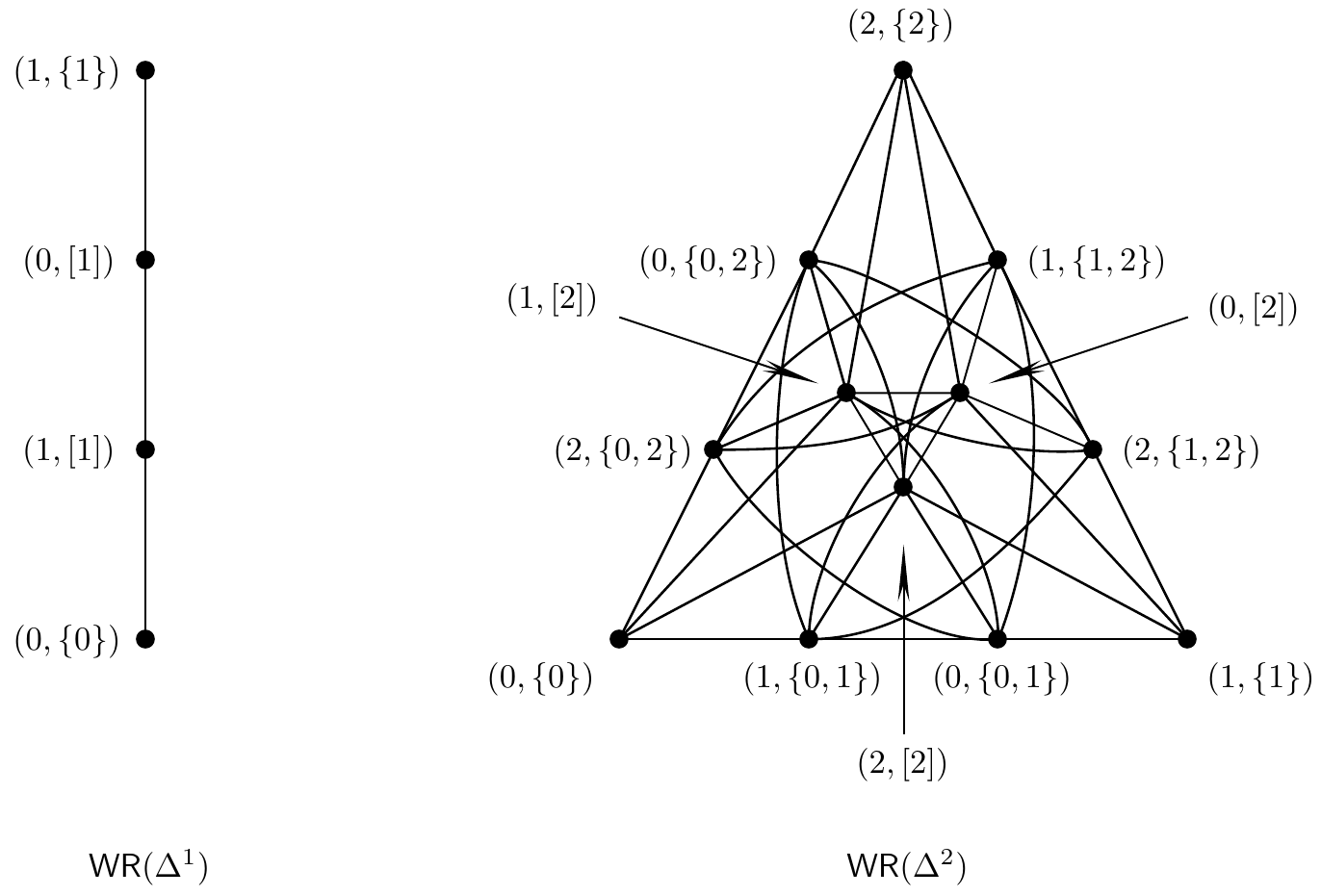}
\end{center}
\caption{Protocol Complex for $n=1$ and $n=2$.}
\label{fig:Protocol_Complex_1_2}
\end{figure}
Lemma~\ref{lem:E_l+1_E_l} implies that every maximal simplex of $\mathsf{WR}_{l+1}$ is a simplex of $\mathsf{WR}_l$, which implies that
$\mathsf{WR}_{l+1}$ is a subcomplex of $\mathsf{WR}_{l}$. From now on $\sigma$ will denote a simplex of $\mathsf{WR}_l$. For $0\leq k\leq l$ let
$\sigma_k^{<}=\{(i,view_i)\in\sigma \ : \ |view_i|<n+1-k\}.$ In a similar way  $\sigma_k^{=}$ and $\sigma_k=\sigma_k^{<}\cup\sigma_k^{=}$
are defined. Therefore, the set of processes in $\sigma$ which participate in the $l+1$ call recursive of algorithm \ref{alg:IS} is partitioned in those which read $n+1-l$ processes and those which read less than $n+1-l$ processes in the $l+1$ layer shared memory.
Let us define
$\mathtt{I}_{\sigma}^<=\bigcup\limits_{i\in\mathtt{Id}(\sigma_l^<)}view_i$  and  $\mathtt{I}_{\sigma}=\bigcup\limits_{i\in\mathtt{Id}(\sigma_l)}view_i$.

\begin{theorem}\label{theo:not_in_WR_l+1}
$\sigma\in\mathsf{WR}_{l+1}$ if only if $\mathtt{I}_{\sigma}^<\cap\mathtt{Id}(\sigma_l^=)=\emptyset$ and $|\mathtt{I}_{\sigma}^<|<n+1-l$.
\end{theorem}

\begin{proof}
Suppose $\mathtt{I}_{\sigma}^<\cap\mathtt{Id}(\sigma_l^=)\neq\emptyset$ or $|\mathtt{I}_{\sigma}^<|=n+1-l$. Then there exists a $IS$-execution $\alpha=\alpha_0,\ldots,\alpha_{l+1}$ such that $\sigma_l^<\subseteq view(\alpha_{l+1})$. In addition there exist processes $i$ and $k$ such that $|view_i|<n+1-l$, $|view_k|=n+1-l$ and $k\in view_i$. This implies that $k$ wrote in the $l+1$ shared memory,  a contradiction. For the other direction, since $\mathtt{I}_{\sigma}^<\cap\mathtt{Id}(\sigma_l^{=})=\emptyset$ and $|\mathtt{I}_{\sigma}^<|<n+1-l$ we can build an $IS$-execution $\alpha=\alpha_0,\ldots,\alpha_{l+1}$ such that $\sigma\subseteq view(\alpha)$.
\qed
\end{proof}

Notice that $\mathtt{I}_{\sigma}$ represents the set of processes which have been read in the $l+1$  recursive call of the algorithm in Figure~\ref{alg:IS}.

\begin{corollary}\label{cor:properties_WR_l}
If $\sigma\not\in\mathsf{WR}_{l+1}$ then
\begin{multicols}{2}
\begin{enumerate}
\item $|\mathtt{I}_{\sigma}|=n+1-l$.
\item $\mathtt{I}_{\sigma}=\mathtt{I}_{\tau}$ for all $\sigma\subseteq\tau$.
\end{enumerate}
\end{multicols}
\end{corollary}

Let $inv(\sigma)=\{(k,\mathtt{I}_{\sigma}) \ : \ k\in\mathtt{I}_{\sigma}\backslash\mathtt{I}_{\sigma}^<\}$ if $\mathtt{I}_{\sigma}^<\neq\mathtt{I}_{\sigma}$ else $inv(\sigma)=\{(k,\mathtt{I}_{\sigma}) \ : \ k\in\mathtt{I}_{\sigma}\backslash\mathtt{Id}(\sigma_l^<)\}$. Notice that if $\sigma\not\in\mathsf{WR}_{l+1}$ then $inv(\sigma)\neq\emptyset$.

For the simplices $\sigma^{-}=\sigma-inv(\sigma)$ and $\sigma^{+}=\sigma\cup inv(\sigma)$.

\begin{proposition}\label{prop:properties_simplices_aux}
If $\sigma\not\in\mathsf{WR}_{l+1}$ then
\begin{enumerate}
\item $\sigma^+=\sigma^-\cup inv(\sigma)$.
\item $\sigma^-\subseteq\sigma\subseteq\sigma^+$.
\item $\sigma^-\not\in\mathsf{WR}_{l+1}$.
\item If $\sigma^-\subseteq\tau\subseteq\sigma^+$ then $\sigma^-=\tau^-$ and $\sigma^+=\tau^+$.
\item $(\sigma^-)^-=\sigma^-$.
\end{enumerate}
\end{proposition}

Consider $\mathtt{I}_-^+(\sigma)=\{\tau\in\mathsf{WR}_l \ : \ \sigma^-\subseteq\tau\subseteq\sigma^+\}$. Item $(3)$ above implies that $\mathtt{I}_-^+(\sigma)\cap\mathsf{WR}_{l+1}=\emptyset$ if $\sigma\not\in\mathsf{WR}_{l+1}$. Moreover from $(4)$ it is obtained that $\mathtt{I}_-^+(\sigma)\cap\mathtt{I}_-^+(\tau)=\emptyset$ or $\mathtt{I}_-^+(\sigma)=\mathtt{I}_-^+(\tau)$.

\begin{figure}
\begin{center}
\includegraphics[scale=0.5]{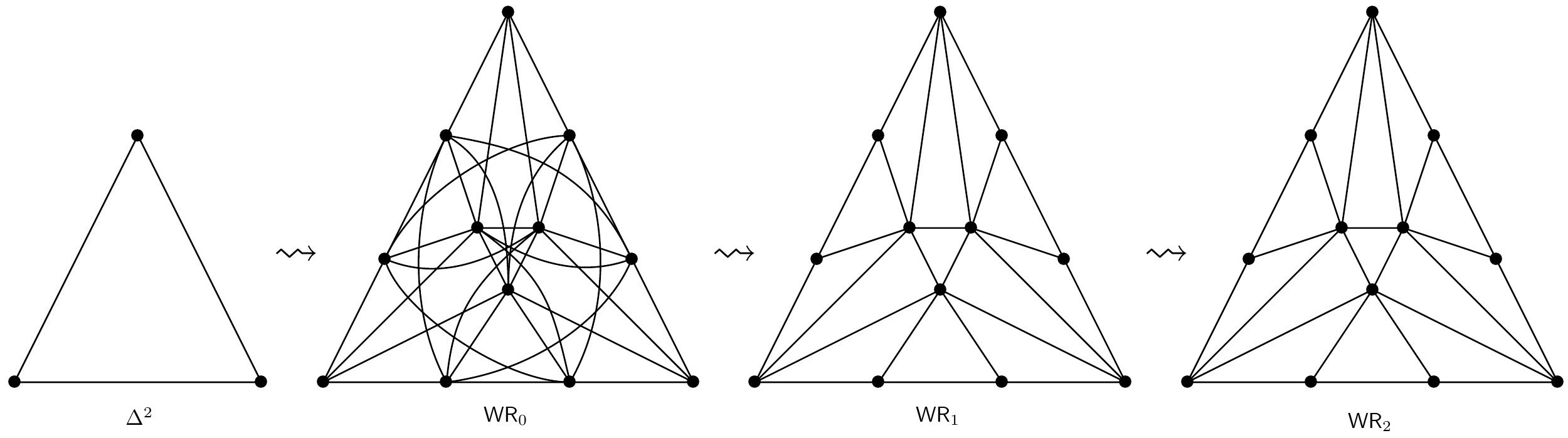}
\end{center}
\caption{Complexes $\mathsf{WR}_l$.}
\label{fig:complexes_algo}
\end{figure}

\vspace{-0.5cm}

\section{Collapsibility}
%of $\mathrm{WR}(\Delta^n)$}
\label{sec:Collapsibility}

Let $S_{[n]}$ denote the permutation group of $[n]$. Notice that if the Id's of processes in a $wr$-execution on $I$ are permuted according to $\pi\in S_{[n]}$ then we obtain a new linear order on $\pi(I)$. In other words if $\alpha$ is a $wr$-execution on $I$ and $\pi\in S_{[n]}$ then $\alpha'=\pi(\alpha)$ is a $wr$-execution on $\pi(I)$. Moreover if $\sigma=view(\alpha)$ then $\pi(\sigma)=view(\pi(\alpha))$. This shows that there exists a natural group action on each simplicial complex $\mathsf{WR}_l$.

\begin{proposition}\label{prop:properties_group_action}
Let $\sigma\in\mathsf{WR}_l$ be a simplex. Then
\begin{multicols}{2}
\begin{enumerate}
\item $\pi(\sigma)\in\mathsf{WR}_l$.
\item $\pi(\sigma^-)=\pi(\sigma)^-$.
\item $\pi(\sigma^+)=\pi(\sigma)^+$.
\item $\pi(\mathtt{I}_-^+(\sigma))=\mathtt{I}_-^+(\pi(\sigma))$.
\end{enumerate}
\end{multicols}
\end{proposition}

For example in Fig.~\ref{fig:Protocol_Complex_1_2},  $\sigma=\{(1,\{1,2\}),(2,\{0,2\}),(0,[2])\}$ and  $\pi(0)=1$, $\pi(1)=2$ and $\pi(2)=0$, then $\pi(\sigma)=\{(2,\{0,2\}),(0,\{0,1\}),(1,[2])\}$.

\begin{theorem}\label{theo:Collapsability}
For every $0\leq l\leq n+1$,
\begin{enumerate}
\item $\mathsf{WR}_l$ is collapsible to $\mathsf{WR}_{l+1}$.
\item $\mathsf{WR}_l$ is $S_{[n]}$-collapsible to $\mathsf{WR}_{l+1}$.
\end{enumerate}
\end{theorem}

\begin{proof}
Since $\sigma\in I_{-}^{+}(\sigma)$ for all simplices $\sigma\in\mathsf{WR}_l$, the intervals $I_{-}^{+}(\sigma)$ cover $L=\{\sigma \ : \ \sigma\in\mathsf{WR}_l, \ \sigma\not\in\mathsf{WR}_{l+1}\}$. Additionally, Proposition~\ref{prop:properties_simplices_aux} $(4)$ implies that $L$ can be decomposed as a disjoint union of intervals
\begin{equation*}
I_{-}^{+}(\sigma_1),\ldots,I_{-}^{+}(\sigma_k)
\end{equation*}
such that $\sigma_i=\sigma_i^+$ for all $1\leq i\leq k$. Suppose $\dim(\sigma_{i+1})_l^<\leq\dim(\sigma_{i})_l^<$ or if $\dim(\sigma_{i+1})_l^<=\dim(\sigma_{i})_l^<$ then $\dim(\sigma_{i+1})\leq\dim(\sigma_{i})$. We will prove by induction on $i$, $1\leq i\leq k$, that
\begin{equation*}
\mathsf{WR}_l^{i}\searrow_{\sigma_i^-}\mathsf{WR}_l^{i+1}
\end{equation*}
where $\mathsf{WR}_l^{1}=\mathsf{WR}_l$ and $\mathsf{WR}_l^{k+1}=\mathsf{WR}_{l+1}$. If there exists a maximal simplex $\sigma\in\mathsf{WR}_l^{i}$ such that $\sigma_i\subseteq\sigma$ then $\sigma=\sigma_j$ for some $i\leq j\leq k$. Hence $(\sigma_i)_l^<\subseteq(\sigma_j)_l^<$ and therefore $\sigma_i=\sigma_j$. Now suppose there exists a maximal simplex $\sigma_j\in\mathsf{WR}_l^{i}$ with $i\leq j\leq k$ such that $\sigma_i^-\subseteq\sigma_j$. This implies that $(\sigma_i)_l^<=(\sigma_j)_l^<$ and $inv(\sigma_i)=inv(\sigma_j)$. Thus $\sigma_i=\sigma_i^-\cup inv(\sigma_i)\subseteq\sigma_j\cup inv(\sigma_j)=\sigma_j$
and therefore $\sigma_i^-$ is free in $\mathsf{WR}_l^{i}$. Therefore,
\begin{equation*}
\mathsf{WR}_l=\mathsf{WR}_l^{1}\searrow_{\sigma_1^-}\ldots\searrow_{\sigma_k^-}\mathsf{WR}_l^{k+1}=\mathsf{WR}_{l+1}.
\end{equation*}
Now if we specify in  more detail the order of the sequence, the complex $\mathsf{WR}_l$ can be collapsed to $\mathsf{WR}_{l+1}$ in a $S_{[n]}$-equivariant way. First note that if $\pi(\sigma_i)\in\mathtt{I}_-^+(\sigma_j)$ for some $1\leq j\leq k$, then Proposition~\ref{prop:properties_group_action} $(3)$ and Proposition~\ref{prop:properties_simplices_aux} $(4)$ imply that $\pi(\sigma_i)=\sigma_j$. Moreover, $\dim(\sigma_i)_l^<=\dim\pi(\sigma_i)_l^<$ and $\dim(\sigma_i)=\dim\pi(\sigma_i)$. Hence the set $\{\sigma_1,\ldots,\sigma_k\}$ can be partitioned according to the equivalence relation given by: $\sigma_i\sim\sigma_j$ if there exists $\pi\in S_{[n]}$ such that $\pi(\sigma_i)=\sigma_j$. Let $\tau_1,\ldots,\tau_p$ be representatives of the equivalence classes which satisfy the order given in the proof of the item $1$, then
\begin{equation*}
\mathsf{WR}_l\searrow_{S_{[n]}(\tau_1^-)}\cdots\searrow_{S_{[n]}(\tau_p^-)}\mathsf{WR}_{l+1}.
\end{equation*}
\qed
\end{proof}

Fig.~\ref{fig:complexes_collapse} illustrates the collapsing procedure $\mathsf{WR}_0\searrow_{S_{[n]}}\mathsf{WR}_1$ for $n=2$. In this case consider the simplexes $\sigma_1=\{(1,\{1,2\}),(2,\{0,2\}),(0,[2])\}$ and $\sigma_2=\{(0,\{0,1\}),(1,[2]),(2,[2])\}$ then
\begin{equation*}
\mathsf{WR}_0^{1}\searrow_{S_{[n]}(\sigma_1^-)}\mathsf{WR}_0^{2}\searrow_{S_{[n]}(\sigma_2^-)}\mathsf{WR}_0^{3}=\mathsf{WR}_1
\end{equation*}

\begin{figure}
\begin{center}
\includegraphics[scale=0.5]{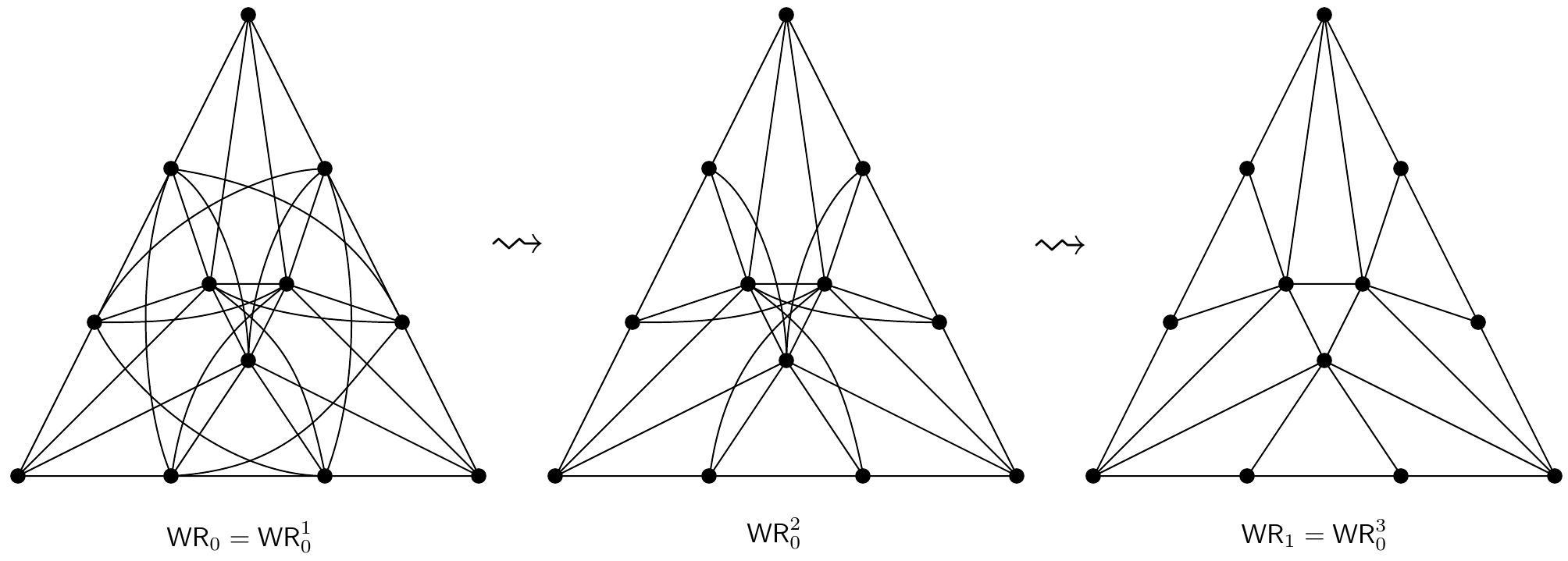}
\end{center}
\caption{$S_{[n]}$-collapse.}
\label{fig:complexes_collapse}
\end{figure}

We have the following consequence.

\begin{corollary}\label{cor:collapsibility_ComplexIS}
For every natural number $n$, the simplicial complex $\mathsf{WR}(\Delta^n)$ is $S_{[n]}$-collapsible to $\chi(\Delta^n)$.
\end{corollary}

\subsubsection*{Multi-round protocol complex}

%Note that we can see the one-round write/read model as a map which assigns to each input configuration $\sigma$ the set of all executions with entry $\sigma$.
A \emph{carrier map}
%is an important map in distributed computing since it models the behavior of the iterated immediate snapshot model and iterated write/read model. A carrier map
$\Phi$ from complex  $\mathcal{C}$ to  complex $\mathcal{D}$  assigns to each simplex $\sigma$ a subcomplex $\Phi(\sigma)$ of $\mathcal{D}$ such that $\Phi(\tau)\subseteq\Phi(\sigma)$ if $\tau\subseteq\sigma$. 
%Thus we can see the protocol complex $\mathsf{WR}(\Delta^n)$ as a carrier map.
%(see appendix \ref{app:muti-round}). 
The protocol complex  of  the iterated write/read model (see~Fig. \ref{fig:iterated_complex}),  $k\geq 0$, is
%\begin{equation*}
$\mathsf{WR}^{(k+1)}(\Delta^n)=\bigcup\limits_{\sigma\in\mathsf{WR}^{(k)}(\Delta^n)}\mathsf{WR}(\sigma).
$
%\end{equation*}

\begin{figure}
\begin{center}
\includegraphics[scale=0.48]{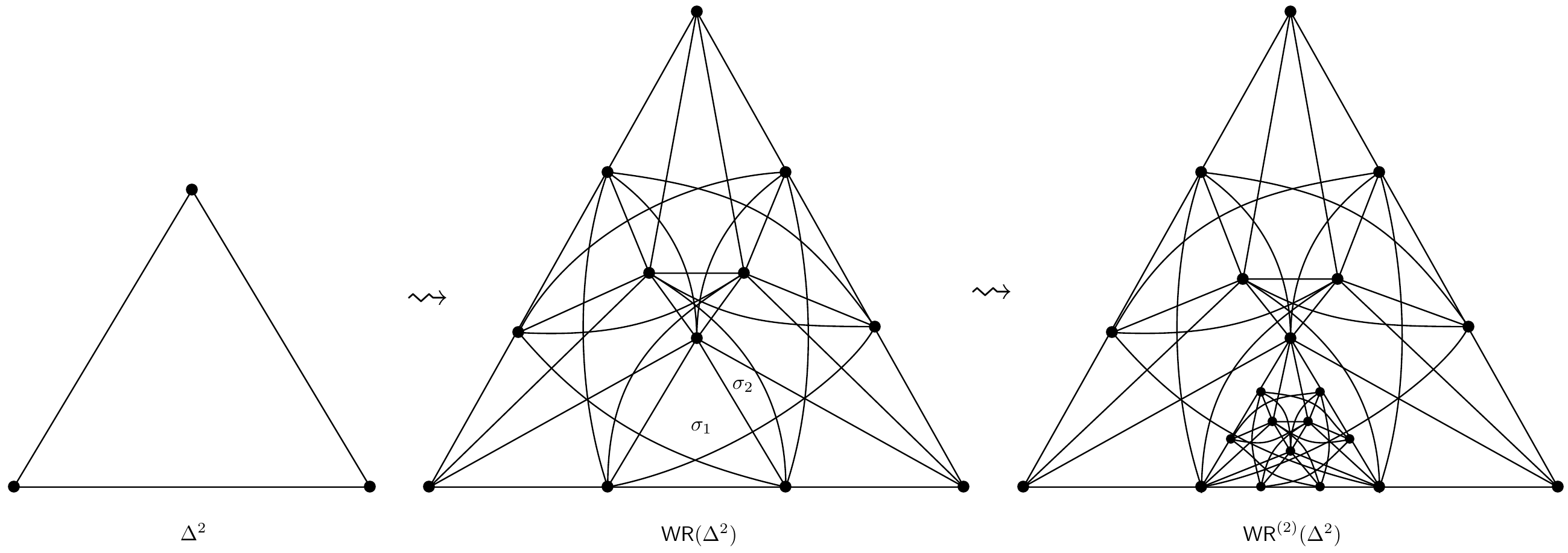}
\end{center}
\caption{Complexes of the iterated model;  in $\mathsf{WR}^{(2)}(\Delta^2)$ only 
$\mathsf{WR}(\sigma_1)$ is depicted.}
\label{fig:iterated_complex}
\end{figure}

\begin{corollary}\label{cor:collasability_IteratedComplexWR}
For all $k\geq 1$,
$\mathsf{WR}^{(k)}(\Delta^n)\searrow\chi^{(k)}(\Delta^n).$
\end{corollary}

%From the protocol complex $\mathsf{WR}(\Delta^2)$ we  apply the carrier map $\mathsf{WR}$ to the maximal simplex $\sigma_1$ to describe the second-round protocol complex $\mathsf{WR}^{(2)}(\Delta^2)$, see Fig. \ref{fig:iterated_complex}.
The collapsing procedure 
%of corollary \ref{cor:collasability_IteratedComplexWR} 
consists first in collapsing, in  parallel, each subcomplex $\mathsf{WR}(\sigma)$ where $\sigma$ is a maximal simplex of $\mathsf{WR}^{(k-1)}(\Delta^n)$ as in   Theorem~\ref{theo:Collapsability}. An illustration is in  
Fig.~\ref{fig:first_collapsing}, applied to the  simplexes 
 $\sigma_1=\{(0,\{0,1\}),(1,\{0,1\}),(2,[2])\}$ and $\sigma_2=\{(0,\{0,1\}),(1,[2]),(2,[2])\}$ of $\mathsf{WR}(\Delta^2)$. 
Second, 
%since $\mathsf{WR}^{(k-1)}(\Delta^n)$ is collapsed to $\chi^{(k-1)}(\Delta^n)$, we use this collapsing procedure to 
we collapse $\chi(\mathsf{WR}^{(k-1)}(\Delta^n))$ to $\chi^{(k)}(\Delta^n)$. 
%Not omitted from Ext Abst
For example $\sigma_2$ is collapsed to $\tau=\{(0,\{0,1\}),(1,[2])\}$, see Fig. \ref{fig:second_collapsing}, then 
Proposition~\ref{prop:CollapsabilityChromatic} implies that $\chi(\sigma_1)$ is collapsed to $\chi(\tau)$.

\begin{figure}
\begin{center}
\includegraphics[scale=0.4]{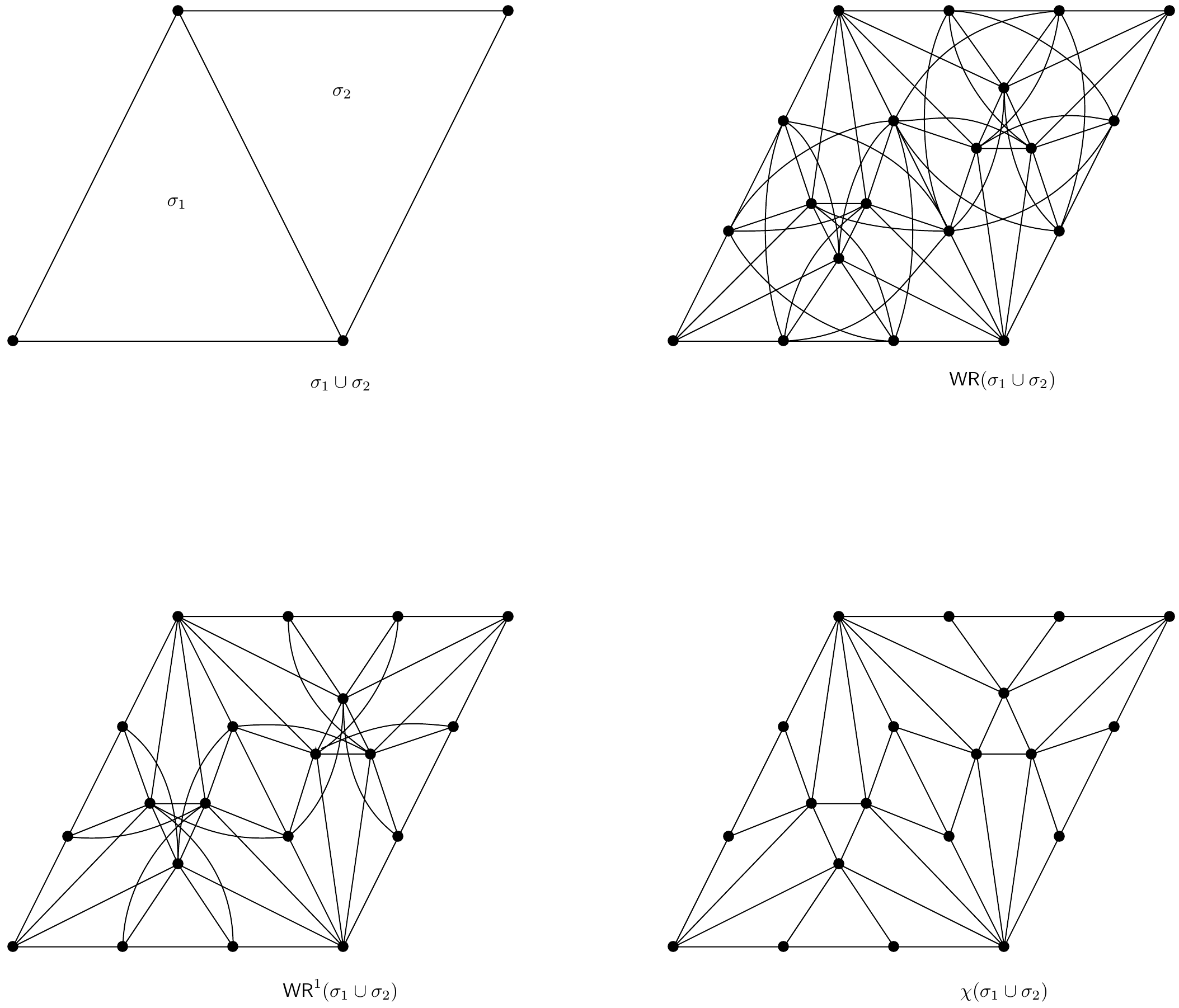}
\end{center}
\caption{First Collapsing.}
\label{fig:first_collapsing}
\end{figure}

%Not omitted from Ext Abst
\begin{figure}
\begin{center}
\includegraphics[scale=0.45]{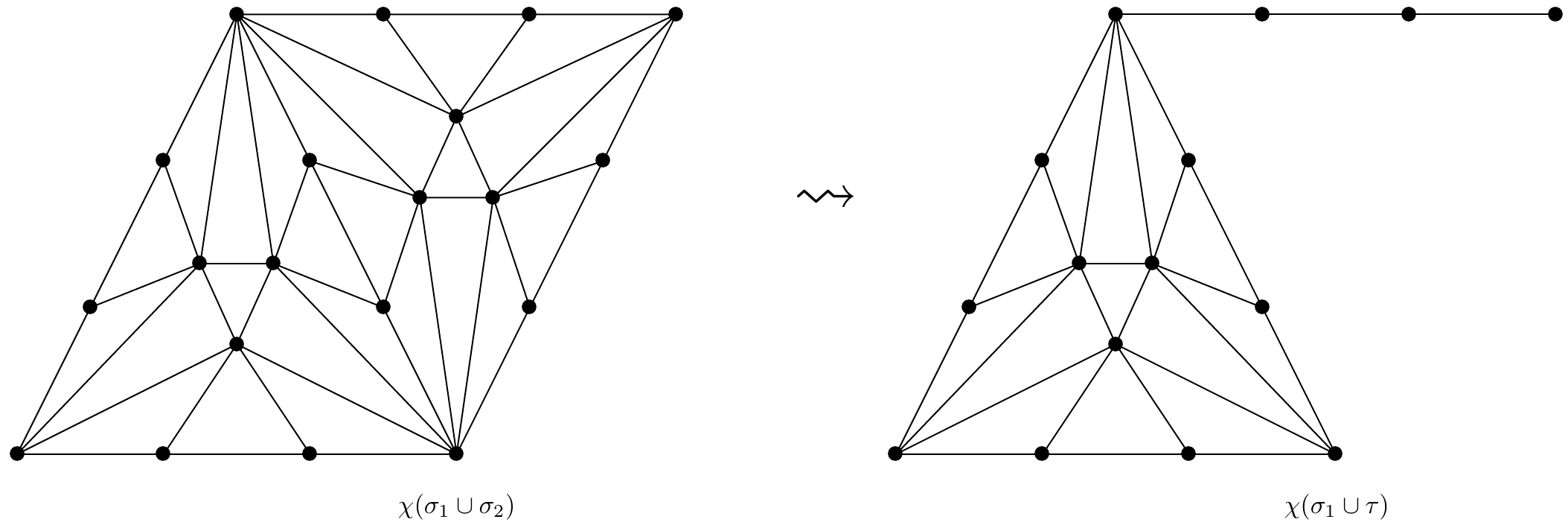}
\end{center}
\caption{Second Collapsing.}
\label{fig:second_collapsing}
\end{figure}

\bibliography{bibtexdoctorado}

\begin{thebibliography}{10}

\bibitem{1990AtomicSnapshotSharedMemory_AADGM}
Yehuda Afek, Hagit Attiya, Danny Dolev, Eli Gafni, Michael Merritt, and Nir
  Shavit.
\newblock Atomic snapshots of shared memory.
\newblock {\em J. ACM}, 40(4):873--890, Sep. 1993.

\bibitem{AttiyaR2002}
Hagit Attiya and Sergio Rajsbaum.
\newblock The combinatorial structure of wait-free solvable tasks.
\newblock {\em SIAM J. Comput.}, 31(4):1286--1313, April 2002.

\bibitem{BiranMZ90}
Ofer Biran, Shlomo Moran, and Shmuel Zaks.
\newblock {A Combinatorial Characterization of the Distributed 1-Solvable
  Tasks}.
\newblock {\em J. Algorithms}, 11(3):420--440, 1990.

\bibitem{1993GeneralizedFLPImposibility_BG}
Elizabeth Borowsky and Eli Gafni.
\newblock Generalized flp impossibility result for t-resilient asynchronous
  computations.
\newblock In {\em Proc. 25th Annual ACM Symp. on Theory of Computing}, STOC,
  pages 91--100, New York, NY, USA, 1993. ACM.

\bibitem{1993ImmediateAtomicSnapshot_BG}
Elizabeth Borowsky and Eli Gafni.
\newblock Immediate atomic snapshots and fast renaming.
\newblock In {\em Proc. 12th ACM Symp. on Principles of Distributed Computing},
  PODC, pages 41--51, New York, NY, USA, 1993. ACM.

\bibitem{1997SimpleAlgorithmicallyReasoned_BG}
Elizabeth Borowsky and Eli Gafni.
\newblock A simple algorithmically reasoned characterization of wait-free
  computation (extended abstract).
\newblock In {\em Proceedings of the Sixteenth Annual ACM Symposium on
  Principles of Distributed Computing}, PODC '97, pages 189--198, New York, NY,
  USA, 1997. ACM.

\bibitem{BorowskyGLR01}
Elizabeth Borowsky, Eli Gafni, Nancy Lynch, and Sergio Rajsbaum.
\newblock The {BG} distributed simulation algorithm.
\newblock {\em Distributed Computing}, 14(3):127--146, 2001.

\bibitem{FischerLP85}
M.~Fischer, N.~A. Lynch, and M.~S. Paterson.
\newblock {Impossibility Of Distributed Commit With One Faulty Process}.
\newblock {\em Journal of the ACM}, 32(2), April 1985.

\bibitem{2010RecursionDictributedComputing_GR}
Eli Gafni and Sergio Rajsbaum.
\newblock Recursion in distributed computing.
\newblock In Shlomi Dolev, Jorge Cobb, Michael Fischer, and Moti Yung, editors,
  {\em Stabilization, Safety, and Security of Distributed Systems}, volume 6366
  of {\em Lecture Notes in Computer Science}, pages 362--376. Springer Berlin
  Heidelberg, 2010.

\bibitem{2014IteratedChromaticCollapsible_GMT}
Eric Goubault, Samuel Mimram, and Christine Tasson.
\newblock Iterated chromatic subdivisions are collapsible.
\newblock {\em Applied Categorical Structures}, pages 1--42, 2014.

\bibitem{2000CompObstructions_Havlicek}
John Havlicek.
\newblock Computable obstructions to wait-free computability.
\newblock {\em Distributed Computing}, 13(2):59--83, 2000.

\bibitem{2004NoteHomotyTypeASComplex_Havlicek}
John Havlicek.
\newblock A note on the homotopy type of wait-free atomic snapshot protocol
  complexes.
\newblock {\em SIAM J. Comput.}, 33(5):1215--1222, 2004.

\bibitem{2013DistributedCombinatorialTopology_HKR}
Maurice Herlihy, Dmitry Kozlov, and Sergio Rajsbaum.
\newblock {\em {Distributed Computing Through Combinatorial Topology}}.
\newblock Elsevier, Imprint Morgan Kaufmann, 2013.

\bibitem{HerlihyR12}
Maurice Herlihy and Sergio Rajsbaum.
\newblock {Simulations and reductions for colorless tasks}.
\newblock In {\em Proceedings of the 2012 ACM symposium on Principles of
  distributed computing}, PODC '12, pages 253--260, New York, NY, USA, 2012.
  ACM.

\bibitem{1999TopologicalStructureAsynchronous_HS}
Maurice Herlihy and Nir Shavit.
\newblock The topological structure of asynchronous computability.
\newblock {\em J. ACM}, 46(6):858--923, November 1999.

\bibitem{HerlihyShavitBook2008}
Maurice Herlihy and Nir Shavit.
\newblock {\em The Art of Multiprocessor Programming}.
\newblock Morgan Kaufmann Publishers Inc., San Francisco, CA, USA, 2008.

\bibitem{HoestS97}
Gunnar Hoest and Nir Shavit.
\newblock {Towards a topological characterization of asynchronous complexity}.
\newblock In {\em Proc. 16th ACM Symp. Principles of distributed computing},
  PODC, pages 199--208, New York, NY, USA, 1997. ACM.

\bibitem{2008SimplicialComplexesGraphs_J}
Jakob. Jonsson.
\newblock {\em Simplicial Complexes of Graphs. 10.1007/978-3-540-75859-4}.
\newblock Lecture Notes in Mathematics,. Springer Berlin Heidelberg,, Berlin,
  Heidelberg :, 2008.

\bibitem{2012ChromaticSubdivision_K}
Dmitry~N. Kozlov.
\newblock Chromatic subdivision of a simplicial complex.
\newblock {\em Homology Homotopy Appl.}, 14(2):197--209, 2012.

\bibitem{2014TopologyISComplex_K}
Dmitry~N. Kozlov.
\newblock Topology of the immediate snapshot complexes.
\newblock {\em Topology Appl.}, 178:160--184, 2014.

\bibitem{2013TopologyComplexView_K}
Dmitry~N. Kozlov.
\newblock Topology of the view complex.
\newblock {\em Homology Homotopy Appl.}, 17(1):307--319, 2015.

\bibitem{LouiAA:87}
M.~C. Loui and H.~H. Abu-Amara.
\newblock {\em {Memory requirements for agreement among unreliable asynchronous
  processes}}, volume~4, pages 163--183.
\newblock JAI press, 1987.

\bibitem{2008IteratedRestrictedImmdediate_RRT}
Sergio Rajsbaum, Michel Raynal, and Corentin Travers.
\newblock The iterated restricted immediate snapshot model.
\newblock In {\em COCOON}, volume 5092 of {\em Lecture Notes in Computer
  Science}, pages 487--497. Springer, 2008.

\bibitem{SaksZ00}
Michael Saks and Fotios Zaharoglou.
\newblock {Wait-Free k-Set Agreement is Impossible: The Topology of Public
  Knowledge}.
\newblock {\em SIAM J. Comput.}, 29(5):1449--1483, 2000.

\end{thebibliography}

\bibliographystyle{plain}

\newpage

\appendix
\noindent
{\bf\Large Appendix}

\section{Proofs}
\label{app:proofs}

\begin{lemma-repeat}{lem:equal_interval}
Let $\alpha$ and $\alpha'$ be $IS$-executions with $l(\alpha)=l(\alpha')$. Given $0\leq k\leq l$,
\begin{enumerate}
\item If $\alpha\sim\alpha'$ then $\mathtt{Id}(\alpha_k)=\mathtt{Id}(\alpha'_k)$.
\item If $\alpha_k\sim\alpha'_k$ then $\alpha\sim\alpha'$.
\end{enumerate}
\end{lemma-repeat}
\begin{proof}
To see $1$ note that $\mathtt{Id}_k(\alpha_k)=\mathtt{Id}_k(\alpha_k')$ since $view(\alpha)=view(\alpha')$. Hence  $\mathtt{Id}(\alpha_k)=\mathtt{Id}(\alpha'_k)$. Now, $2$ is a consequence of equality $view(\alpha_k)=view(\alpha_k')$.
\qed
\end{proof}

%\begin{lemma-repeat}{lem:E_l+1_E_l}
%Let $\alpha$ be an $IS$-execution with $l(\alpha)=l+2$. Then $view(\alpha)=view(\alpha')$ for some  $IS$-execution $\alpha'$ such that $l(\alpha')=l+1$.
%\end{lemma-repeat}
%\begin{proof}
%Let $\alpha=\alpha_0,\ldots,\alpha_{l+1}$ be an $IS$-execution. Then there exists a $wr$-execution $\alpha'_{l}$ such that
%\begin{equation*}
%view(i,\alpha'_l)=\left\{
%  \begin{array}{ll}
%    view(i,\alpha_{l}), & \hbox{if} \ i\in\mathtt{Id}_{l}(\alpha_l) \\
%    view(i,\alpha_{l+1}), & \hbox{if} \ i\not\in\mathtt{Id}_l(\alpha_l).
%  \end{array}
%\right.
%\end{equation*}
%therefore $\alpha'=\alpha_0,\ldots,\alpha_{l-1},\alpha'_{l}$ satisfies the conditions.
%\qed
%\end{proof}

%\begin{theorem-repeat}{theo:not_in_WR_l+1}
%$\sigma\in\mathsf{WR}_{l+1}$ if only if $\mathtt{I}_{\sigma}^<\cap\mathtt{Id}(\sigma_l^=)=\emptyset$ and $|\mathtt{I}_{\sigma}^<|<n+1-l$.
%\end{theorem-repeat}
%\begin{proof}
%Suppose $\mathtt{I}_{\sigma}^<\cap\mathtt{Id}(\sigma_l^=)\neq\emptyset$ or $|\mathtt{I}_{\sigma}^<|=n+1-l$. Then there exists a $IS$-execution $\alpha=\alpha_0,\ldots,\alpha_{l+1}$ such that $\sigma_l^<\subseteq view(\alpha_{l+1})$. In addition there exist processes $i$ and $k$ such that $|view_i|<n+1-l$, $|view_k|=n+1-l$ and $k\in view_i$. This implies that $k$ wrote in the $l+1$ shared memory,  a contradiction. For the other direction, since $\mathtt{I}_{\sigma}^<\cap\mathtt{Id}(\sigma_l^{=})=\emptyset$ and $|\mathtt{I}_{\sigma}^<|<n+1-l$ we can build an $IS$-execution $\alpha=\alpha_0,\ldots,\alpha_{l+1}$ such that $\sigma\subseteq view(\alpha)$.
%\qed
%\end{proof}

\begin{proposition-repeat}{prop:properties_simplices_aux}
If $\sigma\not\in\mathsf{WR}_{l+1}$ then
\begin{enumerate}
\item $\sigma^+=\sigma^-\cup inv(\sigma)$.
\item $\sigma^-\subseteq\sigma\subseteq\sigma^+$.
\item $\sigma^-\not\in\mathsf{WR}_{l+1}$.
\item If $\sigma^-\subseteq\tau\subseteq\sigma^+$ then $\sigma^-=\tau^-$ and $\sigma^+=\tau^+$.
\item $(\sigma^-)^-=\sigma^-$.
\end{enumerate}
\end{proposition-repeat}
\begin{proof}
Items $1$ and $2$ follow directly from the definitions of $\sigma^-$ and $\sigma^+$. Now, for $3$, notice that $\sigma_l^<=(\sigma^-)_l^<$, which implies
$\mathtt{I}_{\sigma}^<=\mathtt{I}_{\sigma^-}^<$. Therefore %$\mathtt{Id}(\sigma^-)_l^<\neq\mathtt{I}_{\sigma^-}^<$ and
the statement follows from Theorem~\ref{theo:not_in_WR_l+1}. To prove $4$, first note that $\tau=\sigma^-\cup\rho$, where $\rho\subseteq inv(\sigma)$.  Therefore $\sigma_l^<=\tau_l^<$ and $\mathtt{I}_{\sigma}=\mathtt{I}_{\tau}$ so $inv(\sigma)=inv(\tau)$. This implies that
\begin{equation*}
\sigma^-=(\sigma-inv(\sigma))-inv(\sigma)\subseteq\tau-inv(\tau)\subseteq(\sigma\cup inv(\sigma))-inv(\sigma)=\sigma^-
\end{equation*}
then $\sigma^-=\tau^-$. From $1$ it is obtained
\begin{equation*}
\tau^+=\tau^-\cup inv(\tau)=\sigma^-\cup inv(\sigma)=\sigma^+.
\end{equation*}
Finally, to prove $(5)$, observe that $(\sigma^-)_l^<=\sigma_l^<$ and $\mathtt{I}_{\sigma^-}=\mathtt{I}_{\sigma}$. Thus $inv(\sigma^-)=inv(\sigma)$, which proves
\begin{equation*}
(\sigma^-)^-=\sigma^--inv(\sigma^-)=(\sigma-inv(\sigma))-inv(\sigma)=\sigma-inv(\sigma)=\sigma^-.
\end{equation*}
\qed
\end{proof}

\begin{proposition-repeat}{prop:properties_group_action}
Let $\sigma\in\mathsf{WR}_l$ be a simplex.
\begin{enumerate}
\item $\pi(\sigma)\in\mathsf{WR}_l$.
\item $\pi(\sigma^-)=\pi(\sigma)^-$.
\item $\pi(\sigma^+)=\pi(\sigma)^+$.
\item $\pi(\mathtt{I}_-^+(\sigma))=\mathtt{I}_-^+(\pi(\sigma))$.
\end{enumerate}
\end{proposition-repeat}
\begin{proof}
Item $(1)$ is clear. $(2)$ and $(3)$ are deduced from the equalities $\pi(\sigma_l^<)=\pi(\sigma)_l^<$ and $\pi(\sigma_l^=)=\pi(\sigma)_l^=$. The last item follows from Proposition~\ref{prop:properties_simplices_aux} $(4)$.
\qed
\end{proof}

Let $\Delta, \Gamma$ be simplicial complexes. Consider the set $C=\{\tau_1,\ldots,\tau_k\}$ of free simplexes in $\Delta$ such that $I(\tau_i)\cap I(\tau_j)=\emptyset$ for all $i\neq j$. This equation means that we can collapse $\tau_1,\ldots,\tau_k$ in a parallel way. This will be denoted by $\Delta\searrow_{C}\Gamma$.

\begin{corollary-repeat}{cor:collasability_IteratedComplexWR}
For all $k\geq 1$,
$\mathsf{WR}^{(k)}(\Delta^n)\searrow\chi^{(k)}(\Delta^n).$
\end{corollary-repeat}
\begin{proof}
The proof is by induction on $k$. Theorem~\ref{theo:Collapsability} proves the base case $k=1$. For $k>1$, let $\sigma_1,\ldots,\sigma_m$ be the maximal simplexes of $\mathsf{WR}^{(k-1)}(\Delta^n)$. Now consider $\sigma_1$ as a standard simplicial complex, then let $\tau_1^{-},\ldots,\tau_k^{-}$ be a sequence of the $S_{[n]}$-collapsability
\begin{equation*}
\mathsf{WR}_l(\sigma_1)\searrow_{S_{[n]}}\mathsf{WR}_{l+1}(\sigma_1)
\end{equation*}
with $0\leq l\leq n-1$ where $\mathsf{WR}_0(\sigma_1)=\mathsf{WR}(\sigma_1)$ and $\mathsf{WR}_n(\sigma_1)=\chi(\sigma_1)$. On the other hand there exists a isomorphism $f_i$ from $\mathsf{WR}(\sigma_1)$ to $\mathsf{WR}(\sigma_i)$ for all $1\leq i\leq m$. This implies that $f_i(\tau_1^{-}),\ldots,f_i(\tau_k^{-})$ is a sequence of the $S_{[n]}$-collapsability
\begin{equation*}
\mathsf{WR}_l(\sigma_i)\searrow_{S_{[n]}}\mathsf{WR}_{l+1}(\sigma_i)
\end{equation*}
for all $1\leq i\leq m$. Consider $G=\{f_1,\ldots,f_m\}$ with $f_1$ the identity simplicial map and $C_j=S_{[n]}(G(\tau_j^{-}))$ with $1\leq j\leq k$. Then for all $0\leq l\leq n-1$
\begin{equation*}
\mathsf{WR}^{(k)}_l(\Delta^n)\searrow_{C_1}\cdots\searrow_{C_k}\mathsf{WR}^{(k)}_{l+1}(\Delta^n)
\end{equation*}
where $\mathsf{WR}^{(k)}_0(\Delta^n)=\mathsf{WR}^{(k)}(\Delta^n)$ and $\mathsf{WR}^{(k)}_n(\Delta^n)=\chi(\mathsf{WR}^{(k-1)}(\Delta^n))$. Thus
\begin{equation*}
\mathsf{WR}^{(k)}(\Delta^n)\searrow\chi(\mathsf{WR}^{(k-1)}(\Delta^n)).
\end{equation*}
By induction hypothesis
\begin{equation*}
\mathsf{WR}^{(k-1)}(\Delta^n)\searrow\chi^{(k-1)}(\Delta^n).
\end{equation*}
Then proposition \ref{prop:CollapsabilityChromatic} implies that
\begin{equation*}
\chi(\mathsf{WR}^{(k-1)}(\Delta^n))\searrow\chi(\chi^{(k-1)}(\Delta^n))
\end{equation*}
and therefore
\begin{equation*}
\mathsf{WR}^{(k)}(\Delta^n)\searrow\chi^{(k)}(\Delta^n).
\end{equation*}
\qed
\end{proof}

\section{Chromatic Subdivision}
\label{app:chromatic-subdivision}

In the proof of Corollary~\ref{cor:collasability_IteratedComplexWR} we used  Proposition~\ref{prop:CollapsabilityChromatic} which appears in~\cite{2014IteratedChromaticCollapsible_GMT}.  For completeness,  we present the proof in  our notation.

Let $\sigma$ be a simplex of $\chi(\Delta^n)$. Notice there exists a vertex $(i,view_i)\in\sigma_0^<$ such that $view_j\subseteq view_i$ for all $(j,view_j)\in\sigma_0^<$. Then $\sigma^c$ will denote the simplex
\begin{equation*}
\sigma^c=\sigma\cup\{(j,[n]) \ : \ j\in[n]\backslash view_i\}.
\end{equation*}

\begin{lemma}\label{lem:CollapsabilityChromatic}
Let $n$ be a natural number. Then $\chi(\Delta^n)$ is collapsible.
\end{lemma}
\begin{proof}
For $1\leq k\leq n+1$ let
\begin{equation*}
\Delta_k^n=\Delta_{k-1}^n\backslash \bigcup\limits_{\sigma\in\Sigma_{k-1}} I(\sigma_0^<)
\end{equation*}
with $\Delta_0^n=\chi(\Delta^n)$ and $\Sigma_{k-1}=\{\sigma\in\Delta_{k-1}^n \ : \ \dim \sigma_0^<=n-k\}$. Then the simplex $\sigma_0^<$ is a free face of $\sigma^c$ in the simplicial complex $\Delta_{k-1}^n$. Hence
\begin{equation*}
\chi(\Delta^n)=\Delta_{0}^n\searrow\cdots\searrow\Delta_{n+1}^n
\end{equation*}
where $\Delta_{n+1}$ is the void complex.
\qed
\end{proof}

For each natural number $n$ and $p\in[n]$ let
\begin{equation*}
\Lambda_p^n=\Delta^n\backslash I(\sigma')
\end{equation*}
where $\sigma'=[n]\backslash\{p\}$.

\begin{proposition}\label{prop:CollapsabilityChromatic}
For all natural numbers $n$
\begin{equation*}
\chi(\Delta^n)\searrow\chi(\Lambda_p^n).
\end{equation*}
\end{proposition}
\begin{proof}
The collapsing procedure consists in three stages.

\noindent{\bf Stage $1$}. Let
\begin{equation*}
\Sigma=\{\sigma\in\chi(\Delta^n) \ : \ \sigma_0^==\{(p,[n])\}\}.
\end{equation*}
Given $1\leq k\leq n$, consider the simplicial complex
\begin{equation*}
K_k^n=K_{k-1}^n\backslash\bigcup\limits_{\sigma\in\Sigma_{k-1}}I(\sigma_0^<)
\end{equation*}
with $K_0^n=\chi(\Delta^n)$ and $\Sigma_{k-1}=\{\sigma\in\Sigma\cap K_{k-1}^n \ : \ \sigma \hbox{ is maximal}\}$. Notice $\sigma_0^<$ is a free face of $\sigma$ since $\dim \sigma=\dim \sigma_0^<+1$. Therefore
\begin{equation*}
\chi(\Delta^n)=K_0^n\searrow\cdots\searrow K_n^n.
\end{equation*}

\noindent{\bf Stage $2$}.
Now consider the simplexes
\begin{equation*}
\Sigma'=\{\sigma\in K_n^n \ : \ (p,[n])\in\sigma\}.
\end{equation*}
Given $1\leq k\leq n$ let
\begin{equation*}
T_k^n=T_{k-1}^n\backslash\bigcup\limits_{\sigma\in\Sigma'_{k-1}}I(\sigma|p)
\end{equation*}
where $\Sigma_{k-1}'=\{\sigma\in\Sigma'\cap T_{k-1}^n \ : \ \dim\sigma_0^<=n-k-1\}$ and $\sigma|p=\sigma_0^<\cup\{(p,[n])\}$. Then $\sigma|p$ is a free face of $\sigma^c$ with $\sigma\in\Sigma'$. Hence,
\begin{equation*}
K_n^n=T_0^n\searrow\cdots\searrow T_n^n.
\end{equation*}

\noindent{\bf Stage $3$}. For $k$, $1\leq k\leq n$, let
\begin{equation*}
M_k=M_{k-1}\backslash\bigcup\limits_{\sigma\in\Sigma_{k-1}''}I(\sigma_0^<)
\end{equation*}
where $M_0=T_n^n$ and $\Sigma_{k-1}''=\{\sigma\in K_n^n \ : \ \dim\sigma_0^==n-k\}$. Now suppose $\sigma_1,\ldots,\sigma_l$ are the simplexes such that $\sigma=(\sigma_i)_0^=$ and $\dim(\sigma_i)_0^==n-k$. Notice if $\sigma\in\Sigma_{k-1}''$ then there exist simplexes $\tau$ and $\sigma_i$ such that $\sigma=\tau\cup\sigma_i$ with $\tau\in\chi(\Delta^{I_i})$ and $I_i=[n]\backslash\mathtt{Id}(\sigma_i)$. For each $i$, $1\leq i\leq l$, let $\tau_1^{(i)},\ldots,\tau_m^{(i)}$ be the sequence of collapses of $\chi(\Delta^{I_i})$ given in the Lemma \ref{lem:CollapsabilityChromatic}. Therefore
\begin{equation*}
\tau_1^{(1)}\cup\sigma_1,\ldots,\tau_m^{(1)}\cup\sigma_1,\ldots,\tau_1^{(l)}\cup\sigma_l,\ldots,\tau_m^{(l)}\cup\sigma_l
\end{equation*}
is a sequence of collapses leading from $M_{k-1}$ to $M_k$. Hence
\begin{equation*}
M_0\searrow\cdots\searrow M_n=\chi(\Lambda_p^n).
\end{equation*}
\qed
\end{proof}

\end{document}